\documentclass[sn-mathphys-num]{sn-jnl}

\usepackage{graphicx}%
\usepackage{multirow}%
\usepackage{amsmath,amssymb,amsfonts}%
\usepackage{amsthm}%
\usepackage{mathrsfs}%
\usepackage[title]{appendix}%
\usepackage{xcolor}%
\usepackage{textcomp}%
\usepackage{manyfoot}%
\usepackage{booktabs}%
\usepackage{algorithmic}
\usepackage[linesnumbered,ruled,vlined]{algorithm2e} 

\theoremstyle{thmstyleone}%
\newtheorem{theorem}{Theorem}[section]
\newtheorem{lemma}[theorem]{Lemma}
\newtheorem{corollary}[theorem]{Corollary}
\newtheorem{claim}[theorem]{Claim}

\theoremstyle{thmstyletwo}%
\newtheorem{remark}{Remark}%

\theoremstyle{thmstylethree}%

\begin{document}

\title[Solving No-wait Scheduling for TSN with Daisy-chain Topology]{Solving No-wait Scheduling for Time-Sensitive Networks with Daisy-Chain Topology}


\author[1]{\fnm{Qian} \sur{Li}}

\author[2]{\fnm{Henan} \sur{Liu}}

\author[4]{\fnm{Heng} \sur{Liu}}

\author[2,3]{\fnm{Yuyi} \sur{Wang}}

\affil[1]{\orgdiv{Shenzhen lnternational  Center For Industrial  And  Applied  Mathematics}, \orgname{Shenzhen Research Institute of Big Data}, \orgaddress{\city{Shenzhen}, \country{China}}}

\affil[2]{\orgname{CRRC, Zhuzhou Institute}, \orgaddress{\city{Zhuzhou}, \country{China}}}

\affil[3]{\orgname{Tengen Intelligence Institute}, \orgaddress{\city{Zhuzhou}, \country{China}}}

\affil[4]{\orgname{South China University of Technology}, \orgaddress{\city{Guangzhou}, \country{China}}}


\abstract{Time-Sensitive Networking (TSN) is a set of standards aiming to enable deterministic and predictable communication over Ethernet networks. However, as the standards of TSN do not specify how to schedule the data streams, the main open problem around TSN is how to compute schedules efficiently and effectively. In this paper, we solve this open problem for no-wait schedules on the daisy-chain topology, one of the most commonly used topologies. Precisely, we develop an efficient algorithm that optimally computes no-wait schedules for the daisy-chain topology, with a time complexity that scales polynomially in both the number of streams and the network size.

The basic idea is to recast the no-wait scheduling problem as a variant of a graph coloring problem where some restrictions are imposed on the colors available for every vertex, and where the underlying graph is an interval graph. Our main technical part is to show that this variant of graph coloring problem can be solved in polynomial time for interval graphs, though it is NP-hard for general graphs. Evaluations based on real-life TSN systems demonstrate its optimality and its ability to scale with up to tens of thousands of streams.}




\newcommand{\warn}[1]{{\color{red}{#1}}}
\newcommand{\supp}{\mathrm{supp}}
\newcommand{\sw}{\mathrm{SW}}
\newcommand{\edg}{\mathcal{E}}
\newcommand{\interval}{\mathcal{I}}
\newcommand{\clique}{c}
\renewcommand{\vec}[1]{\boldsymbol{#1}}
\newcommand{\subalgo}{\mathsf{Find}}
\newcommand{\kmax}{k^\ast}

\maketitle

\section{Introduction}
In recent years, the rapid advancement of communication technologies has revolutionized various industries, enabling the development of complex systems with demanding real-time requirements. From industrial automation, in-vehicle communication, and avionics to multimedia streaming and telecommunication networks, there is an increasing need for reliable and deterministic communication to ensure the timely delivery of critical data. In particular, one central characterization of Industry 4.0 paradigm is networked cyber-physical systems, where computers control physical processes. So a real-time communication network with deterministically bounded network delay and jitter is usually needed to guarantee the physical system under control. 

\emph{Time-Sensitive Networking} (TSN) has emerged as a set of standards to address these requirements and provide a unified solution for time-sensitive applications over Ethernet networks.
TSN represents a significant evolution of traditional Ethernet, which was primarily designed for best-effort communication without any guarantees on timing or determinism. TSN is a suite of standards developed by the Institute of Electrical and Electronics Engineers (IEEE) 802.1 working group, and its main objective is to enable deterministic and predictable communication over Ethernet networks. To achieve the objective, TSN introduces several features and enhancements to traditional Ethernet, including time synchronization, traffic shaping and scheduling, and stream reservation. 

More specifically, TSN incorporates precise clock synchronization protocols, such as IEEE 1588 \cite{ieee1588} or IEEE802.1AS \cite{ieee802}, to achieve tight synchronization (up to nanoseconds precision) of switches and end stations in a network, which ensures that all the devices are able to coordinate their actions and thus allows scheduling of streams across the switched network. End stations inject data streams into the switched network at precise predetermined time points based on a schedule. Then streams pass through switches and reach their respective destinations. The destination of a stream must be an end station. The routing paths of streams are a predefined part of the input and are fixed. A switch supports multiple priority queues per egress port. Each queue has a so-called gate, and whether the gate is opening or closed determines whether this queue can access the medium.  The opening and closing of the gates of an egress port are controlled by a so-called Gate Control List (GCL), which is also predetermined according to a schedule.

However, the standards of TSN do not configure the scheduling. In fact, it is the main open problem around TSN how to compute a schedule that coordinates the injection times for all streams and the GCLs of switches to ensure the requested real-time requirements of all time-triggered traffic (a.k.a. scheduled traffic) are met. It turns out that deciding whether there is a valid schedule for a given set of streams is a very difficult combinatorial optimization problem: it is NP-hard even when restricted to various special classes of instances \cite{survey}. The scheduling algorithms in the literature can be classified into exact approaches and heuristic approaches. The exact approaches express the scheduling problem as a satisfiability modulo theory (SMT) problem, an integer linear programming (ILP) problem, or a constraint programming (CP) problem, etc, and then invoke the corresponding solver to find the optimal solution. The exact approaches would compute an optimal schedule if one exists, but they are not scalable beyond very small problem instances. Besides the exact approaches, many heuristic algorithms have been developed to try to find reasonably good schedules within a short time, such as the algorithms based on Tabu search \cite{durr2016no} or simulated annealing \cite{Reusch2021DependabilityAwareRA}. However, in the common case, the heuristic approaches cannot deduce whether a problem instance is infeasible, nor are guaranteed to find a solution if one exists.  See e.g. \cite{survey} for an excellent survey of scheduling algorithms in TSN. In summary, the existing scheduling algorithms are either poorly scalable or suboptimal.

In contrast to previous works designing general-purposed scheduling algorithms but without theoretical guarantees, this paper takes a different approach: it aims to develop algorithms with theoretical guarantees for a special kind of scenarios, which however still covers a notable proportion of real-life scenarios. Precisely, we focus on \emph{no-wait} scheduling of time-triggered streams in TSN with \emph{daisy-chain} topology.

\begin{figure}[t]
	\centering
	\includegraphics[scale=0.55]{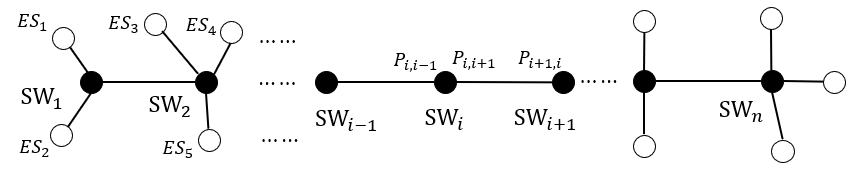}
	\caption{Illustration of the daisy-chain topology. Filled circles represent switches, and hollow ones represent end stations.}
	\label{fig:chain}
\end{figure}

The daisy-chain topology (a.k.a. the line topology, see Figure \ref{fig:chain} for an illustration) is one of the most commonly used topologies and also the simplest and cheapest topologies to implement, and can be found in several real-life TSN systems such as the ones equipped on trains.

 In a no-wait schedule, streams are not allowed to wait in switches, and a switch must forward a stream to the next node immediately upon receiving the stream. The rationale of the no-wait constraint is that  both the network delay and jitter of all 
streams are minimized simultaneously since no queuing delay occurs. No-wait schedules are particularly suitable for extremely high real-time requirements. Durr and Nayak \cite{durr2016no} presented a Tabu Search algorithm to compute no-wait schedules, where they model the no-wait scheduling problem as the well-known no-wait job-shop scheduling problem. Wang et al. \cite{Wang2022DeepRL} proposed a reinforcement learning method for no-wait scheduling, where they train machine learning models aiming to reduce the maximum arrival time among all frames. 
Both of the above two methods are heuristic and thus are suboptimal in the common case. Moreover, they are not scalable beyond medium-size problem instances: the method in \cite{durr2016no} needs about 3 hours to compute schedules for about 1500 streams, and the one in \cite{Wang2022DeepRL} needs about 400 seconds to compute schedules for 100 streams in a network with 9 switches and 10 end stations.
\begin{figure}[t]
	\centering
	\includegraphics[scale=0.6]{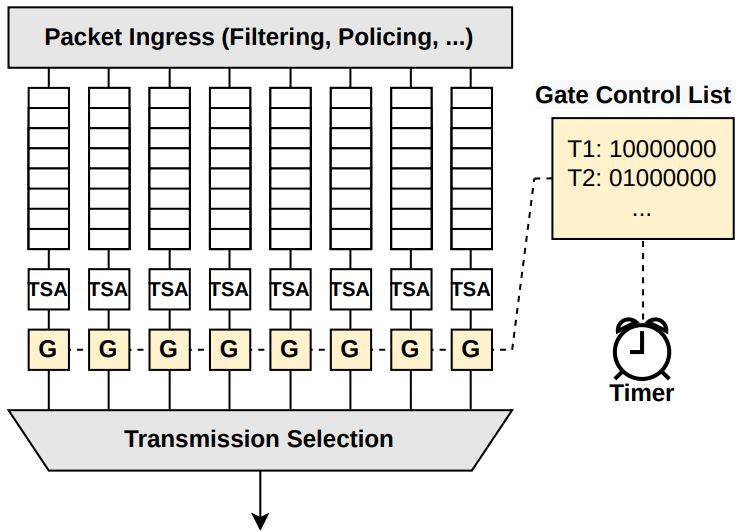}
	\caption{An Illustration of Gate Control List (GCL) from \cite{survey}}
	\label{fig:gcl}
\end{figure}

\vspace{2ex}
\noindent\textbf{Our Contribution.} In this paper, we propose an efficient algorithm (Algorithm \ref{alg:main}) that \emph{optimally} computes no-wait schedules on the daisy-chain topology (a.k.a. the line topology). Precisely, given a set of streams to be transmitted in TSN with a daisy-chain topology, Algorithm \ref{alg:main} returns a no-wait schedule if one exists or proves that no such no-wait schedules exist, in $O((n+|\mathcal{S}|)\cdot |\mathcal{S}|^{1.5}\log_2 p_{\mathcal{S}}+|\mathcal{S}|\cdot p_{\mathcal{S}}\log_2 p_{\mathcal{S}})$ time. Here, $|\mathcal{S}|$ is the number of streams, $n$ is the length of the daisy chain, and $p_{\mathcal{S}}$ is the least common multiple of the periods of streams in $\mathcal{S}$ (a.k.a. hyperperiod of $\mathcal{S}$). Evaluations based on a real-life TSN network (with 32 switches) show that Algorithm \ref{alg:main} can compute schedules for about 45,000 flows in only about 30 minutes. Moreover, if we only care about whether there exist no-wait schedules or not, then it can be decided in $O(|\mathcal{S}|\cdot n)$ time (Corollary \ref{cor:decide}). 

Our basic idea is that the no-wait scheduling problem can be recast as a variant of a graph coloring problem where some restrictions are imposed on the colors available for every vertex. Since the topology is the daisy chain, the underlying graph is an \emph{interval graph}. Our main technical part is to show that this variant of graph coloring problem can be solved in polynomial time for interval graphs, though it is NP-hard for general graphs.

\vspace{2ex}
\noindent\textbf{Organization.} The rest of this paper is organized as follows. We provide some details about TSN relevant to this paper in Section \ref{sec:model}, and then some notations and the mathematical model in Section \ref{sec:notation}. In Section \ref{sec:alg}, we present the algorithms and the analysis. We present evaluation results in Section \ref{sec:evaluation}, and conclude the paper in Section \ref{sec:conclusion}.

\section{System Model}\label{sec:model}
According to the IEEE 802.1Qbv standards, a Time-Sensitive Network consists of network elements (i.e., switches and end stations) and a Central Network Controller (CNC) for configuring and managing the network. In this paper, we focus on scheduling time-triggered traffic (a.k.a. scheduled traffic), which consists of periodic data streams with hard real-time requirements. A time-triggered stream is specified by its source device, destination device, routing path, period, and amount of data per period. These properties of all streams are a part of the input to scheduling algorithms. Besides, the source device and destination device of each stream must be end stations.

The transit of time-trigger streams in TSN is scheduled by specifying (i) the \emph{injection time} of each stream and (ii) the \emph{Gate Control List} (GCL) of each egress port of switches. Precisely, the schedule is executed periodically for an indefinite number of times. The injection time of a stream is the relative time to the start of the scheduled period to be injected into the network. Each egress port of a switch has multiple queues, and each queue has a so-called gate. When the gate is in the open state, data in the queue is considered for transmission. The opening and closing of the gates of an egress port are controlled by a GCL. A GCL entry consists of two parts: a time interval $[T_i,T_{i+1}]$ , and a bit string indicating which gates are open or closed in the time interval $[T_i,T_{i+1}]$. Typically, the GCLs of all egress ports are programmed with the same period as the schedule. See Figure \ref{fig:gcl} for an illustration of GCL. In addition, in contrast with egress ports, an ingress port is capable of receiving multiple streams simultaneously.

Scheduling would be simplified when we restrict our attention to no-wait scheduling: we only need to specify the injection time of each stream. Precisely, noting that streams cannot be scheduled to wait at closed gates, we can deploy any no-wait schedule by opening all gates all the time.

In this paper, we assume that for all streams and all switches, the time that the switch forwards the stream (i.e., from when the switch receives the first bit of the stream to when the next node receives the first bit) is the same and set to be 1 unit time. The forwarding time consists of transmission time and several types of delays, namely the propagation delay of signals along the link, the processing delay for deciding on which port to forward an incoming packet, and the queuing delay of a packet in the queue of an outgoing port (which is zero in no-wait schedules). The forwarding time interval of two streams by the same egress port could overlap slightly, but the transmission time interval must not since otherwise their electrical signals would interfere. Here, we impose that the forwarding time interval of two streams by the same egress port must not overlap. This imposed restriction will make things much simpler, with decreasing little capacity of TSN because all the delays are much smaller than the transmission time. Finally, we remark that the assumption of the same forwarding time is reasonable, since (a) the forwarding time is mainly determined by the packet size and the link bandwidth; (b) the packet sizes and the bandwidth of links respectively are usually similar in practical TSN; and (c) large packets can be split into multiple small unit-forwarding-time packets.

\begin{figure}[ht]
	\centering
	\includegraphics[scale=0.65]{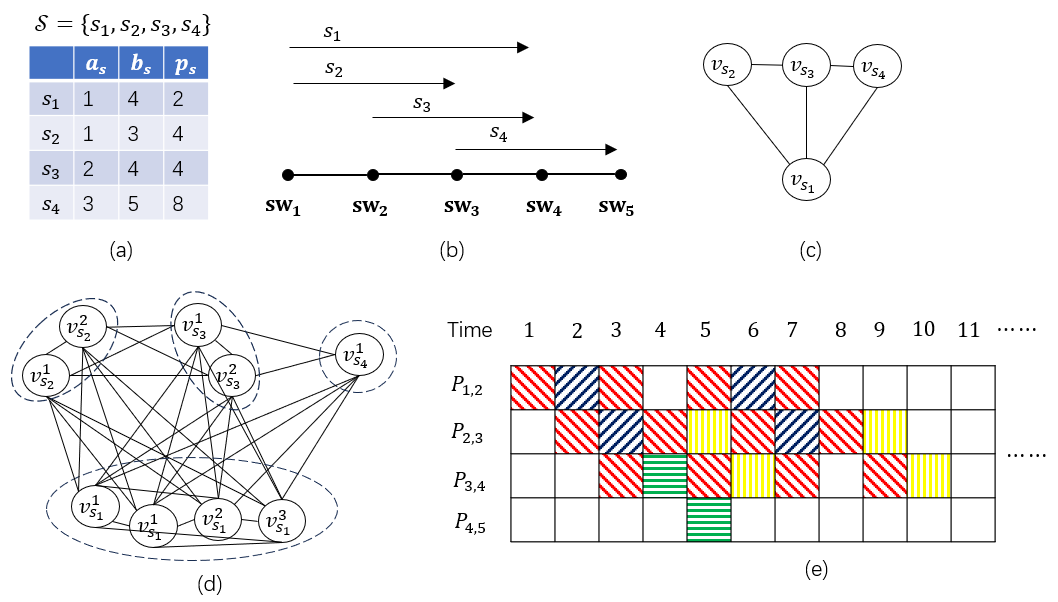}
	\includegraphics[scale=0.65]{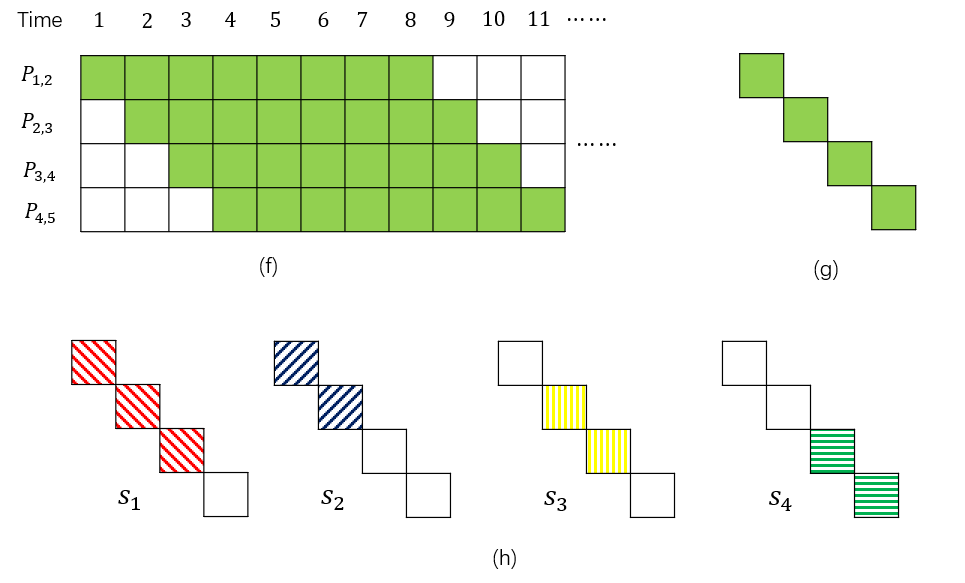}
	\caption{(a) a set of streams $\mathcal{S}$; (b) routing paths; (c) $G_{\mathcal{S}}$; (d) $G^\star_{\mathcal{S}}$; (e) a  Gantt chart describing a no-wait schedule; (f) a parallelogram (in green color); (g) a layer; (h) the Gantt chart description of each stream}
	\label{fig:ill}
\end{figure}

\section{Notations and Mathematical Modeling}\label{sec:notation}
In this section, we present the mathematical model of the no-wait scheduling problem for TSN with a daisy-chain topology and recast it to a variant of graph coloring problem with some restrictions on the colors available for every vertex (see Theorem \ref{thm:equiv}).

\subsection{Notations}
For two integers $a\leq b$, define $[a,b]:=\{a,a+1\cdots,b\}$. As shorthand, we write $[n]$ for $[1,n]=\{1,2,\cdots,n\}$. Figure \ref{fig:chain} illustrates the daisy-chain topology. In Figure \ref{fig:chain}, there are $n$ switches wired together in sequence, and named by $\sw_1,\sw_2,\cdots,\sw_n$ from left to right. We use $P_{i,i+1}$ to denote the right egress port of $\sw_i$, which sends data from $\sw_i$ to $\sw_{i+1}$, and use $P_{i,i-1}$ to denote the left egress port of $\sw_i$, which sends data from $\sw_i$ to $\sw_{i-1}$. Let $\edg[i]$ denote the set of end stations directly connected with $\sw_i$. For example, in Figure \ref{fig:chain}, $\edg[1]=\{ES_1,ES_2\}$ and $\edg[2]=\{ES_3,ES_4,ES_5\}$.  

We use $\mathcal{S}$ to denote the set of data streams to be scheduled. For a stream $s\in \mathcal{S}$, let $p_s\in\mathbb{N}$ denote the period of $s$ (i.e, the period of $s$ is $p_s$ units time), and define $a_s,b_s\in[n]$ such that $s$ is sent from an end station in $\edg[a_s]$ to another one in $\edg[b_s]$. We focus on streams where $a_s\neq b_s$. According to whether $a_s>b_s$ or not, we divide $\mathcal{S}$ into two groups:
\begin{itemize}
\item[(i)] one group consists of all streams $s\in\mathcal{S}$ that goes from right to left, i.e., $a_s>b_s$; and 
\item[(ii)] the other group consists of all streams $s\in\mathcal{S}$ that goes in the reverse direction, i.e., $a_s<b_s$. 
\end{itemize}
Transmission of these two groups involves different egress ports and thus can be scheduled separately. So, for convenience of presentation, we assume $a_s<b_s$ for any $s\in\mathcal{S}$. Note that $s$ will go through the egress ports $P_{a_s,a_s+1},P_{a_s+1,a_s+2},\cdots,P_{b_s-1,b_s}$ in order.

Let $p_{\mathcal{S}}$ denote the hyperperiod of $\mathcal{S}$, which is defined to be the least common multiple $\mathsf{LCM}(p_s:s\in\mathcal{S})$ of the periods of streams in $\mathcal{S}$. A no-wait schedule is executed periodically with the hyperperiod $p_{\mathcal{S}}$ as its period. A schedule for  $\mathcal{S}$ contains $p_{\mathcal{S}}/p_s$ replications of $s$, each having the hyperperiod as its period.

\subsection{Recasting the no-wait scheduling problem as a graph coloring problem} For each stream $s\in \mathcal{S}$, we associate it with an interval $\interval_s:=[a_s,b_s-1]$. Furthermore, we associate $\mathcal{S}$ with an undirected graph $G_{\mathcal{S}}$ which 
\begin{itemize}
\item[(i)] has a vertex $v_s$ for each stream $s\in\mathcal{S}$ and
\item[(ii)] has an edge $(v_s,v_{s'})$ if and only if $\interval_s\cap \interval_{s'}\neq \emptyset$ (i.e., the transmission of $s$ and $s'$ would use a common egress port).
\end{itemize}
A useful observation is that $G_{\mathcal{S}}$ is an \emph{interval graph}.
 In addition, we associate $\mathcal{S}$ with another undirected graph $G^\star_{\mathcal{S}}$ which 
\begin{itemize}
\item[(i)] has $p_{\mathcal{S}}/p_s$ vertices $v_s^1,\cdots,v_s^{p_{\mathcal{S}}/p_s}$ for each stream $s\in\mathcal{S}$ and 
\item[(ii)] has an edge $(v^i_s,v^j_{s'})$ if and only if $\interval_s\cap \interval_{s'}\neq \emptyset$.
\end{itemize}
In other words, $G_{\mathcal{S}}^\star$ is obtained from $G_{\mathcal{S}}$ by duplicating each $v_s$ $p_{\mathcal{S}}/p_s$ times. Note that $G^\star_{\mathcal{S}}$ is also an interval graph. See Figure \ref{fig:ill} (c)-(d) for an example of $G_{\mathcal{S}}$ and $G^\star_{\mathcal{S}}$.

A no-wait scheduling can be graphically described by a Gantt chart (see Figure \ref{fig:ill} (e) for an example). Specifically, in the Gantt chart, 
\begin{itemize}
\item The horizontal axis corresponds to time and is divided into slots of 1 unit time. Recall that 1 unit time is the amount of time that a switch forwards a stream.
\item The vertical axis corresponds to the $n-1$ egress ports $P_{1,2},P_{2,3},\cdots,P_{n-1,n}$.
\item A stream $s$ corresponds to $b_s-a_s$ horizontal bars, each lasting 1 unit time, stepping down from left to right (see Figure \ref{fig:ill} (h) for examples). The ``no-wait" restriction corresponds to that the finishing time of the previous bar is also the beginning time of the next bar.
\end{itemize}
The Gantt chart can be viewed as a $(n-1)\times \infty$ dimensional table. We divide the table into parallelograms in the fashion shown in Figure \ref{fig:ill} (f): each parallelogram consists of $p_\mathcal{S}$ layers, and each layer is defined to be a stepdown stair (Figure \ref{fig:ill} (g) draws a layer). A crucial observation is that the bars of a stream are all contained in only one layer. So the no-wait scheduling is equivalent to packing the given set of stepdown-shape bars (e.g. the ``tiles" in Figure \ref{fig:ill} (h)) into a parallelogram, under the constraint that there is a replication of $s$ between the $((i-1)\cdot p_s+1)$-th layer and the $(i\cdot p_s)$-th layer for all $s$ and $i\in[p_{\mathcal{S}}/p_s]$. 
For example, Figure \ref{fig:ill} (e) successfully packs the ``titles" in Figure \ref{fig:ill} (h) into a parallelogram, so a no-wait schedule exists. 

Inspired by the Gantt chart description, the no-wait scheduling can be recast in a variant of graph coloring problem where some restrictions are imposed on the colors available for every vertex. Precisely, recall that a \emph{proper} $q$-coloring of a graph $G=(V,E)$ is a function $C:V\rightarrow [q]$ such that no two adjacent vertices share the same color, i.e., $C(v)\neq C(v')$ for any $(v,v')\in E$. For a proper $p_{\mathcal{S}}$-coloring of $G_{\mathcal{S}}^\star$, we say it \emph{good} if  for each $s\in\mathcal{S}$ and each $i\in[p_{\mathcal{S}}/p_s]$, $C(v_s^i)\in [(i-1)\cdot p_s+1,i\cdot p_s]$. Suppose $G_{\mathcal{S}}^\star$ admits a good $p_{\mathcal{S}}$-coloring $C$. Then we can get a no-wait schedule by placing the $i$-th replication of $s$ in the $C(v_s^i)$-th layer of the parallelogram. On the other hand, suppose $\mathcal{S}$ admits a no-wait schedule, then we can obtain a good $p_{\mathcal{S}}$-coloring $C$ of $G_{\mathcal{S}}^\star$ by setting $C(v_s^i)$ to be the number of the layer containing the $i$-th replication of $s$. In summary, we have the following theorem.
\begin{theorem}\label{thm:equiv}
There exists a no-wait schedule for $\mathcal{S}$ if and only if there is a good $p_{\mathcal{S}}$-coloring of $G_{\mathcal{S}}^\star$. Furthermore, a no-wait schedule can be easily obtained from a good $p_{\mathcal{S}}$-coloring.
\end{theorem}

\section{Algorithm Design}\label{sec:alg}
By Theorem \ref{thm:equiv}, the no-wait scheduling problem is equivalent to finding a good $p_{\mathcal{S}}$-coloring of $G_{\mathcal{S}}^\star$. In this section, we present an efficient algorithm to solve this variant of graph coloring problem, whose time complexity scales polynomially in both the number of streams and the network size.

\subsection{Baby case: when all streams have the same period}
We first consider the baby case where all streams have the same period, i.e., $p_{s}=p_{\mathcal{S}}$ for any $s\in\mathcal{S}$. In this case, 
$G_{\mathcal{S}}^\star$ reduces to $G_{\mathcal{S}}$, and every proper $p_s$-coloring of $G_{\mathcal{S}}$ is good. So the no-wait scheduling problem further reduces to finding a proper $p_s$-coloring of $G_{\mathcal{S}}$. Recall that $G_\mathcal{S}$ is an interval graph. While the $q$-coloring problem is a famous NP-complete problem for general graphs and general $q$, it can be solved by a greedy algorithm in linear time for interval graphs \cite{interval-graph-k-coloring}. So when all streams have the same period, the no-wait scheduling problem can be solved in linear time, i.e., $O(|\mathcal{S}|\log |\mathcal{S}|)$ time. 

\subsection{General case}
Now, we consider the general case, where streams have different periods. Here, we impose that each $p_s$ must be a power of two, i.e., $p_s=2^k$ for some non-negative integer $k$. On one hand, this imposed restriction does not lose much generality, since any $p_s\in\mathbb{N}$ can be approximated by some $2^k$ up to factor at most $\sqrt{2}\approx 1.414$. On the other hand, it seems that this imposed restriction makes the no-wait scheduling problem tractable: the no-wait scheduling problem can be solved efficiently under this imposed restriction (as we will see later), whereas we are not aware of and haven't come up with any algorithm computing no-wait schedules for general $p_s$ with a time complexity that scales polynomially in both the number of streams and the network size.

In the following, we present our scheduling algorithm. We first introduce some notations. Define $\kmax:=\max_{s\in\mathcal{S}}\log_2 p_s$. Note that $\kmax$ is a non-negative integer and the hyperperiod $p_{\mathcal{S}}=2^{\kmax}$. For $0\leq k\leq \kmax$, let 
\[
\mathcal{S}_{k}=\{s\in\mathcal{S}\mid p_s=2^k\}
\]
denote the subset of $\mathcal{S}$ consisting of all streams with period $2^k$. Define $\mathcal{S}_{\leq k}=\mathcal{S}_1\cup\cdots\cup \mathcal{S}_{k}$. Note that $\mathcal{S}=\mathcal{S}_1\cup \mathcal{S}_2\cup\cdots\cup \mathcal{S}_{\kmax}=\mathcal{S}_{\leq \kmax}$. 

Given a partition $\mathcal{S}_{\kmax}=\mathcal{S}_{\kmax}^a\cup \mathcal{S}_{\kmax}^b$ of $\mathcal{S}_{\kmax}$, where $\mathcal{S}_{\kmax}^a\cap \mathcal{S}_{\kmax}^b=\emptyset$, we define 
\begin{equation}\label{eq:sa}
\mathcal{S}^a:=\mathcal{S}_{\leq \kmax-1}\cup \widehat{\mathcal{S}_{\kmax}^a}
\end{equation}
and 
\begin{equation}\label{eq:sb}
\mathcal{S}^b:=\mathcal{S}_{\leq \kmax-1}\cup \widehat{\mathcal{S}_{\kmax}^b}.
\end{equation}
Here, $\widehat{\mathcal{S}_{\kmax}^a}$ (and $\widehat{\mathcal{S}_{\kmax}^b}$ respectively) is obtained from $\mathcal{S}_{\kmax}^a$ (and $\mathcal{S}_{\kmax}^b$ respectively) by changing the period of its streams from $2^{\kmax}$ to $2^{\kmax-1}$. For example, if $\mathcal{S}_{\kmax}^a$ contains a stream $s$ with $(a_s,b_s,p_s)=(1,4,2^{\kmax})$, then $\widehat{\mathcal{S}_{\kmax}^a}$ contains a stream $\hat{s}$ with $(a_{\hat{s}},b_{\hat{s}},p_{\hat{s}})=(1,4,2^{\kmax-1})$. Note that the hyperperiods of $\mathcal{S}^a$ and $\mathcal{S}^b$ are $2^{\kmax-1}$ rather than $2^{\kmax}$. 

In addition, we can define $G_{\mathcal{S}^a}^\star$ and $G_{\mathcal{S}^b}^\star$. Note that $G_{\mathcal{S}^a}^\star$ and $G_{\mathcal{S}^b}^\star$ have exactly one vertex $v_s^1$ for $s$ in $\widehat{\mathcal{S}_{\kmax}^a}$ or $\widehat{\mathcal{S}_{\kmax}^b}$. We have the following lemma.
\begin{lemma}\label{lem:characterization1}
$G_{\mathcal{S}}^\star$ admits a good $p_{\mathcal{S}}$-coloring if and only if there is a partition $\mathcal{S}_{\kmax}=\mathcal{S}_{\kmax}^a\cup \mathcal{S}_{\kmax}^b$ such that both $G_{\mathcal{S}^a}^\star$ and $G_{\mathcal{S}^b}^\star$ admit good $(p_{\mathcal{S}}/2)$-colorings. 
\end{lemma}
\begin{proof}
On one hand, suppose $C:V\left(G_{\mathcal{S}}^\star\right)\rightarrow \left[2^{\kmax}\right]$ is a good $2^{\kmax}$-coloring of $G_{\mathcal{S}}^\star$. Note that $G_{\mathcal{S}}^\star$ has exactly one vertex $v_s^1$ for each $s\in\mathcal{S}_{\kmax}$. According to whether $C(v_s^1)$ is in $[1,2^{\kmax-1}]$ or $[2^{\kmax-1}+1,2^{\kmax}]$, we get a partition $\mathcal{S}_{\kmax}=\mathcal{S}_{\kmax}^a\cup \mathcal{S}_{\kmax}^b$ where 
\[
\mathcal{S}_{\kmax}^a=\left\{s\in \mathcal{S}_{\kmax}\mid C(v_s^1)\in \left[1,2^{\kmax-1}\right] \right\}
\]
and
\[
\mathcal{S}_{\kmax}^b=\left\{s\in \mathcal{S}_{\kmax}\mid C(v_s^1)\in \left[2^{\kmax-1}+1,2^{\kmax}\right] \right\}.
\] 
For $G_{\mathcal{S}^a}^\star$, we define a coloring $C^a:V(G_{\mathcal{S}^a}^\star)\rightarrow[1,2^{\kmax-1}]$ as follows: for each vertex $v_s^i$ in $G_{\mathcal{S}^a}^\star$, let $C^a(v_s^i)=C(v_s^i)$. One can easily check that $C^a$ is good $2^{\kmax-1}$-coloring of $G_{\mathcal{S}^a}^\star$ by definition. 
For $G_{\mathcal{S}^b}^\star$, we define a coloring $C^b:V(G_{\mathcal{S}^b}^\star)\rightarrow[1,2^{\kmax-1}]$ as follows: 
\begin{itemize}
\item for any $s\in\mathcal{S}_{\leq \kmax-1}$ and any $1\leq i\leq 2^{\kmax-1}/p_s$, define $$C^b(v_s^i)=C\left(v_s^{i+(2^{\kmax-1})/p_s}\right)-2^{\kmax-1};$$
Here, we claim that $C^b(v_s^i)$ is nonnegative and thus properly defined. Since $C$ is a good coloring, by definition, we have that for each $s\in\mathcal{S}$ and each $i\in[p_{\mathcal{S}}/p_s]$, $C(v_s^i)\in [(i-1)\cdot p_s+1,i\cdot p_s]$. In particular, 
$$C\left(v_s^{i+(2^{\kmax-1})/p_s}\right)-2^{\kmax-1}\in[(i-1)\cdot p_s+1,i\cdot p_s].$$

\item for any $s\in\mathcal{S}_{\kmax}^b$, define $C^b(v_s^1)=C(v_s^1)-2^{\kmax-1}$.
\end{itemize}
One can easily check that $C^b$ is good $2^{\kmax-1}$-coloring of $G_{\mathcal{S}^b}^\star$ by definition. 

On the other hand, suppose $\mathcal{S}_{\kmax}=\mathcal{S}_{\kmax}^a\cup \mathcal{S}_{\kmax}^b$ is a partition of $\mathcal{S}_{\kmax}$, and $C^a$ and $C^b$ are good $2^{\kmax-1}$-colorings of $G_{\mathcal{S}^a}^\star$ and $G_{\mathcal{S}^b}^\star$ respectively. We define a coloring of $G_{\mathcal{S}}^\star$ as follows:
\begin{itemize}
\item for any $s\in\mathcal{S}_{\leq \kmax-1}$ and any $1\leq i\leq 2^{\kmax}/p_s$, define 
\begin{align*}
 C(v_s^i)=\begin{cases}
 C^a(v_s^i), & \text{if } 1\leq i\leq \frac{2^{\kmax-1}}{p_s};\\
 C^b\left(v_s^{i-\frac{2^{\kmax-1}}{p_s}}\right)+2^{\kmax-1}, & \text{if } \frac{2^{\kmax-1}}{p_s}+1\leq i\leq \frac{2^{\kmax}}{p_s}.
 \end{cases}
 \end{align*}
\item for any $s\in\mathcal{S}_{\kmax}^a$, define $C(v_s^1)=C^a(v_s^1)$;
\item for any $s\in\mathcal{S}_{\kmax}^b$, define $C(v_s^1)=C^b(v_s^1)+2^{\kmax-1}$.
\end{itemize} 
One can check that $C$ is good $2^{\kmax}$-coloring of $G_{\mathcal{S}}^\star$ by definition. 
\end{proof}
For $0\leq k\leq \kmax$ and $\ell\in[n-1]$, define 
\begin{align*}
\clique_{k,\ell}:=&2^{k-\kmax}\cdot \left|\left\{v_s^i\in V(G_{\mathcal{S}}^\star)\mid s\in\mathcal{S}_{\leq k}\mbox{ and } \ell\in \interval_s\right\}\right|\\
=&\sum_{i=0}^k 2^{k-i}\cdot |\{s\in\mathcal{S}_{i}\mid \ell\in \interval_s\}|.
\end{align*}
In particular, $\clique_{\kmax,\ell}=\left|\left\{v_s^i\in V(G_{\mathcal{S}}^\star)\mid \ell\in \interval_s\right\}\right|$.
The following theorem is the main technical part.
\begin{algorithm}[t]
 	\caption{No-wait scheduling algorithm}\label{alg:main}
    \KwIn{a set of streams $\mathcal{S}$.}
    $\kmax\leftarrow \log_2 p_{\mathcal{S}}$;\\
    \For{each $\ell\in[n-1]$}{
         $\clique_{\kmax,\ell}\leftarrow \sum_{i=0}^{\kmax} 2^{\kmax-i}\cdot |\{s\in\mathcal{S}_{i}\mid \ell\in \interval_s\}|$;\\
         \If{$\clique_{\kmax,\ell}>2^{\kmax}$}{
             Return ``no-wait schedules do not exist";
             }  
    }
    Run the procedure $\subalgo(\mathcal{S},\kmax)$.
\end{algorithm}

\begin{algorithm}[h]
\LinesNumbered
 	\caption{Procedure $\subalgo(\mathcal{S},\kmax)$}\label{alg:procedure}
   \KwIn{an integer $\kmax\geq 0$ and a set of streams $\mathcal{S}$.}
    create an empty dictionary $C$\;
    \If{$\mathcal{S}$ is empty}{
           return the dictionary $C$;
           }
    \If{$\kmax=0$}{
         \For{each $s\in\mathcal{S}$}
             {
             $C(v^1_s)\leftarrow 1$, and add the key-value pair $(v^1_s:C(v^1_s))$ to $C$;
             }
      return the dictionary $C$;  
    }
     \For{each $\ell\in[n-1]$}{
         $\clique_{\kmax-1,\ell}\leftarrow \sum_{i=0}^{\kmax-1} 2^{\kmax-1-i}\cdot |\{s\in\mathcal{S}_{i}\mid \ell\in \interval_s\}|$;
    }
    $\vec{c}\leftarrow \left(2^{\kmax-1}-\clique_{\kmax-1,\ell}:\ell\in[n-1]\right)$\;
     create a $(n-1)\times |\mathcal{S}_{\kmax}|$ matrix $A$\;
    \For{each $\ell\in[n-1]$ and each $s\in\mathcal{S}_{\kmax}$}{
         $A[\ell,s]\leftarrow 1_{\ell\in \mathcal{I}_s}$;
    }
    $\vec{x}^\ast\leftarrow $ an integer solution of $\{A(\vec{1}-\vec{x})\leq\vec{c}, A\vec{x}\leq \vec{c},\vec{0}\leq \vec{x}\leq \vec{1}\}$\;
    $\mathcal{S}_{\kmax}^a\leftarrow \{s\in \mathcal{S}_{\kmax}\mid x^\ast_s=0\}$, and $\mathcal{S}_{\kmax}^b\leftarrow \{s\in \mathcal{S}_{\kmax}\mid x^\ast_s=1\}$\;
    compute $\mathcal{S}^a$ and $\mathcal{S}^b$ as defined in \eqref{eq:sa} and \eqref{eq:sb}\; 
    $C^a\leftarrow \subalgo(\mathcal{S}^a,\kmax-1)$ and $C^b\leftarrow \subalgo(\mathcal{S}^b,\kmax-1)$\;
   \For{each $s\in\mathcal{S}_{\leq \kmax-1}$ and each $1\leq i\leq 2^{\kmax}/p_s$}{
             \If{$1\leq i\leq 2^{\kmax}/(2\cdot p_s)$}{
             $C(v_s^i)\leftarrow C^a(v_s^i)$\;}
         \Else{$C(v_s^i)\leftarrow C^b\left(v_s^{i-2^{\kmax}/(2\cdot p_s)}\right)+2^{\kmax-1}$\;}
         add the key-value pair $(v_s^i,C(v_s^i))$ to dictionary $C$;
     }
    \For{each $s\in\mathcal{S}_{\kmax}^a$}{
         $C(v_s^1)\leftarrow C^a(v_s^1)$, and add the key-value pair $(v_s^1,C(v_s^1))$ to dictionary $C$;}
    \For{each $s\in\mathcal{S}_{\kmax}^b$}{
         $C(v_s^1)\leftarrow C^b(v_s^1)+2^{\kmax-1}$, and add $(v_s^1,C(v_s^1))$ to dictionary $C$;
         }
    return dictionary $C$ as a good $2^{\kmax}$-coloring of $G_{\mathcal{S}}^\star$;
\end{algorithm}

\begin{theorem}\label{thm:main}
$G_{\mathcal{S}}^\star$ admits a good $p_{\mathcal{S}}$-coloring if and only if $\clique_{\kmax,\ell}\leq 2^{\kmax}$ for any $\ell\in[n-1]$.
\end{theorem}
\begin{proof}
The proof is by induction on $\kmax$. When $\kmax=0$, the theorem is trivial. By induction, we assume the theorem holds when $\kmax=k'$, and we are going to show the theorem when $\kmax=k'+1$.

According to Lemma \ref{lem:characterization1}, $G_{\mathcal{S}}^\star$ admits a good $2^{\kmax}$-coloring if and only if there is a partition $\mathcal{S}_{\kmax}=\mathcal{S}_{\kmax}^a\cup \mathcal{S}_{\kmax}^b$ such that both $G_{\mathcal{S}^a}^\star$ and $G_{\mathcal{S}^b}^\star$ admit good $\big(2^{\kmax-1}\big)$-colorings. 
Remember that the hyperperiods of $\mathcal{S}^a$ and $\mathcal{S}^b$ are both $2^{\kmax-1}$. So by the induction hypothesis, $G_{\mathcal{S}^a}^\star$ admit a good $\big(2^{\kmax-1}\big)$-coloring if and only if 
\begin{align}\label{eq:s_a}
c_{\kmax-1,\ell}+|\{s\in\mathcal{S}_{\kmax}^a\mid \ell\in \interval_s\}|\leq 2^{\kmax-1}\mbox{ for any } \ell\in[n-1].
\end{align}
Similarly, by the induction hypothesis, $G_{\mathcal{S}^b}^\star$ admit a good $\big(2^{\kmax-1}\big)$-coloring if and only if
\begin{align}\label{eq:s_b}
c_{\kmax-1,\ell}+|\{s\in\mathcal{S}_{\kmax}^b\mid \ell\in \interval_s\}|\leq 2^{\kmax-1}\mbox{ for any } \ell\in[n-1].
\end{align}

Define a Boolean vector $\vec{x}=(x_s:s\in\mathcal{S}_{\kmax})$ where $x_s=0$ indicates $s\in \mathcal{S}_{\kmax}^a$ and $x_s=1$ indicates $s\in \mathcal{S}_{\kmax}^b$. Define a $(n-1)\times |S_{\kmax}|$ Boolean matrix $A$ where $A[\ell,s]=1_{\ell\in\mathcal{I}_s}$. That is, $A[\ell,s]=1$ if $\ell\in \mathcal{I}_s$ and $A[\ell,s]=0$ otherwise. Note that 
\[
A(\vec{1}-\vec{x})=\left(|\{s\in\mathcal{S}_{\kmax}^a\mid \ell\in \interval_s\}|:\ell\in[n-1]\right)
\]
and
\[
A\vec{x}=\left(|\{s\in\mathcal{S}_{\kmax}^b\mid \ell\in \interval_s\}|:\ell\in[n-1]\right),
\]
where $\vec{1}$ is the all ones vector. Define a vector 
\[
\vec{c}=(2^{\kmax-1}-c_{\kmax-1,\ell}:\ell\in[n-1]).
\]
 Then \eqref{eq:s_a} and \eqref{eq:s_b} can be rewritten as $A(\vec{1}-\vec{x})\leq \vec{c}$ and $A\vec{x}\leq \vec{c}$ respectively. So we have the following claim.
	\begin{claim}\label{claim1}
$G_{\mathcal{S}}^\star$ admits a good $p_{\mathcal{S}}$-coloring if and only if $\{A(\vec{1}-\vec{x})\leq\vec{c}, A\vec{x}\leq \vec{c},\vec{0}\leq \vec{x}\leq \vec{1},\vec{x}\in\mathbb{Z}\}\neq\emptyset$.
	\end{claim}
Noting that $A$ is a consecutive-ones matrix, i.e., the ones appearing consecutively in each column, we have that $A$ is a totally unimodular matrix \cite{fulkerson1965incidence}. So we get the following claim. 
\begin{claim}\label{claim2}
$\{A(\vec{1}-\vec{x})\leq\vec{c}, A\vec{x}\leq \vec{c},\vec{0}\leq \vec{x}\leq \vec{1},\vec{x}\in\mathbb{Z}\}\neq\emptyset$ if and only if $\{A(\vec{1}-\vec{x})\leq\vec{c}, A\vec{x}\leq \vec{c},\vec{0}\leq \vec{x}\leq \vec{1}\}\neq\emptyset$.
\end{claim}
The following claim is crucial.
\begin{claim}\label{claim3}
$\{A(\vec{1}-\vec{x})\leq\vec{c}, A\vec{x}\leq \vec{c},\vec{0}\leq \vec{x}\leq \vec{1}\}\neq\emptyset$ if and only if  $A\cdot \vec{1}\leq 2\vec{c}$.
\end{claim}
\begin{proof}
For convenience of presentation, let $\mathcal{F}:=\{A(\vec{1}-\vec{x})\leq\vec{c}, A\vec{x}\leq \vec{c},\vec{0}\leq \vec{x}\leq \vec{1}\}$.
On one direction, if $A\cdot \vec{1}\leq 2\vec{c}$, then $\frac{1}{2}\cdot\vec{1}\in \mathcal{F}$, so $\mathcal{F}\neq\emptyset$. On the other direction, suppose $\vec{x}^\ast\in \mathcal{F}$. Then $(\vec{1}-\vec{x}^\ast)\in \mathcal{F}$. Because $\mathcal{F}$ is a convex set, $\vec{x}^\ast/2+(\vec{1}-\vec{x}^\ast)/2=\frac{1}{2}\cdot \vec{1}\in \mathcal{F}$, which implies that $A\cdot \vec{1}\leq 2\vec{c}$.
\end{proof}
Finally, $A\cdot \vec{1}\leq 2\vec{c}$ means that 
\begin{align*}
|\{s\in\mathcal{S}_{\kmax}\mid \ell\in \interval_s\}|\leq 2^{\kmax}-2c_{\kmax-1,\ell}\mbox{ for any } \ell\in[n-1],
\end{align*}
which is equivalent to saying: for any $\ell\in[n-1]$, 
\begin{align*}
c_{\kmax,\ell}=2c_{\kmax-1,\ell}+|\{s\in\mathcal{S}_{\kmax}\mid \ell\in \interval_s\}|\leq 2^{\kmax} 
\end{align*}
The conclusion follows from Claims \ref{claim1}, \ref{claim2}, and \ref{claim3} immediately.
\end{proof}
\begin{remark}
It is a basic fact that an interval graph admits a proper $k$-coloring if and only if it contains no clique of size greater than $k$. Thus $G_{\mathcal{S}}^\star$ admits a proper $2^{\kmax}$-coloring if and only if $\clique_{\kmax,\ell}\leq 2^{\kmax}$ for any $\ell\in[n-1]$. So $G_{\mathcal{S}}^\star$ admits a proper $2^{\kmax}$-coloring if and only if it admits a good $2^{\kmax}$-coloring. 
\end{remark}
\begin{figure*}[t]
	\centering
	\includegraphics[scale=0.35]{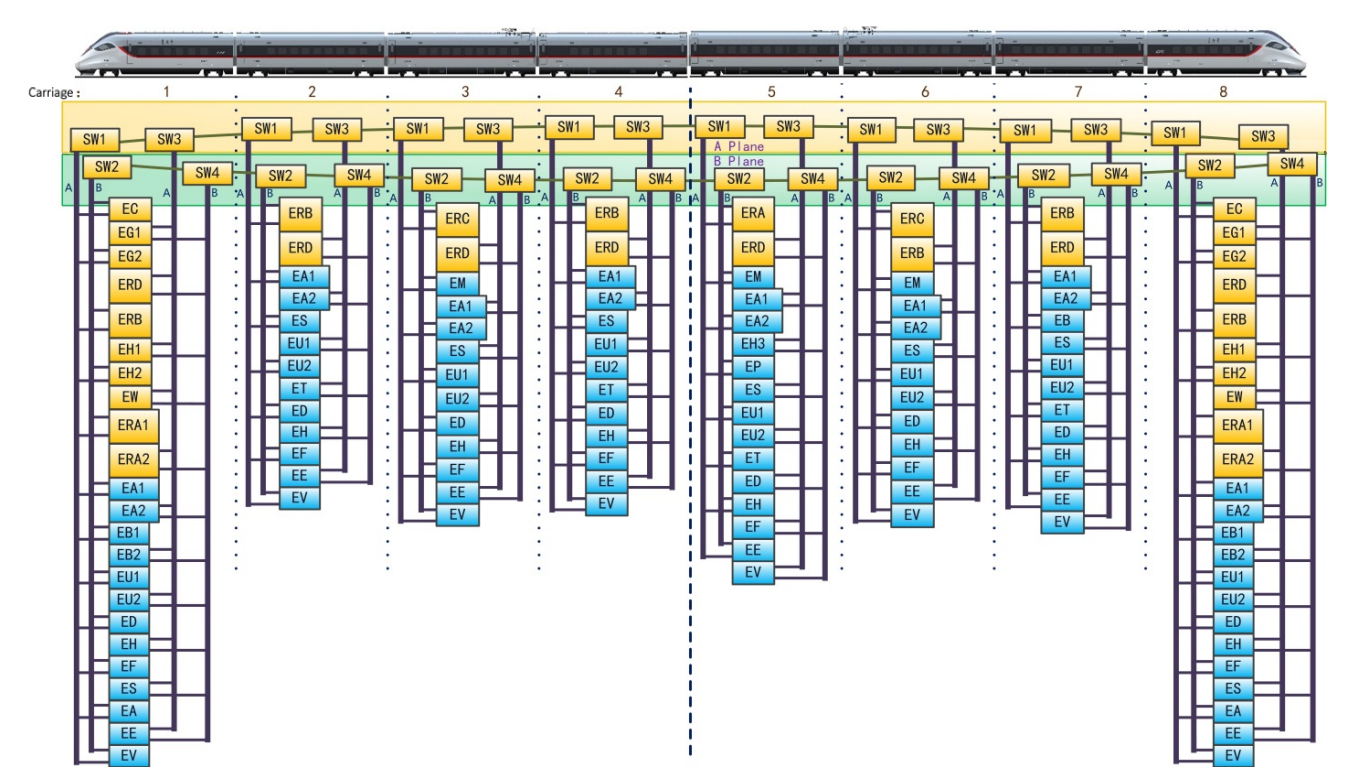}
	\caption{The topology of the communication network on a metro train}
	\label{fig:train}
\end{figure*}
 
Note that all $\clique_{\kmax,\ell}=\sum_{i=0}^{\kmax} 2^{\kmax-i}\cdot |\{s\in\mathcal{S}_{i}\mid \ell\in \interval_s\}|$ can be computed in $O(|\mathcal{S}|\cdot n)$ time. By Theorem \ref{thm:main} and \ref{thm:equiv}, we have 
\begin{corollary}\label{cor:decide}
It can be decided in $O(|\mathcal{S}|\cdot n)$ time whether there exists a no-wait schedule for $\mathcal{S}$. 
\end{corollary}
Furthermore, if no-wait schedule exists for $\mathcal{S}$, then the procedure $\subalgo(\mathcal{S},\kmax)$, described in Algorithm \ref{alg:procedure}, would return a good $2^{\kmax}$-coloring of $G_{\mathcal{S}}^\star$, represented in the form of dictionary $\{(v_s^i,C(v_s^i))\}$. Given a good $2^{\kmax}$-coloring $C:V\left(G_{\mathcal{S}}^\star\right)\rightarrow \left[2^{\kmax}\right]$ of $G_{\mathcal{S}}^\star$, we can easily get a no-wait schedule by placing the $i$-th replication of $s$ in the $C(v_s^i)$-th layer of the parallelogram, i.e., by setting the injection time of the $i$-th replication of $s$ to be $C(v_s^i)+a_s-1$.

\begin{theorem}\label{thm:summary}
Algorithm \ref{alg:main} solves the no-wait scheduling problem in $O((n+|\mathcal{S}|)\cdot |\mathcal{S}|^{1.5}\log_2 p_{\mathcal{S}}+|\mathcal{S}|\cdot p_{\mathcal{S}}\log_2 p_{\mathcal{S}})$ time.
\end{theorem}
\begin{proof}
We first show the correctness of Algorithm \ref{alg:main}. By Theorem \ref{thm:main} and \ref{thm:equiv}, there exist no-wait schedules for $\mathcal{S}$ if and only if $\clique_{\kmax,\ell}\leq 2^{\kmax}$ for any $\ell\in[n-1]$. Moreover, from the proof of Theorem \ref{thm:main},  one can easily verify that the procedure $\subalgo(\mathcal{S},\kmax)$ would return a no-wait schedule if one exists. So Algorithm \ref{alg:main} returns a no-wait schedule if one exists or returns ``no-wait schedules do not exist" otherwise.

What remains is to analyze the time complexity of Algorithm \ref{alg:main}. It executes $O(|\mathcal{S}|\cdot n)$ time to execute Lines 1-5. In the following, we will show that it costs $O((n+|\mathcal{S}|)\cdot |\mathcal{S}|^{1.5}\log_2 p_{\mathcal{S}}+|\mathcal{S}|\cdot p_{\mathcal{S}}\log_2 p_{\mathcal{S}})$ time to run the procedure $\subalgo(\mathcal{S},\kmax)$, and then finish the proof. 

Let $T(|\mathcal{S}|,\kmax)$ denote the time complexity of  $\subalgo(\mathcal{S},\kmax)$. First, it is easy to see that it costs $O(n\cdot |\mathcal{S}|)$ time to execute Lines 1-13, 15-16, and 24-28. It costs $O((n+|\mathcal{S}|)\cdot |\mathcal{S}|^{1.5})$ time to solve the linear programming in Line 14 \cite{vaidya1989speeding}, and costs $O(|\mathcal{S}|\cdot 2^{\kmax})=O(|\mathcal{S}|\cdot p_{\mathcal{S}})$ time to execute Lines 18-23. In addition, it costs at most $2T(|\mathcal{S}|,\kmax-1)$ time to execute Line 17. So we have 
\[
T(|\mathcal{S}|,\kmax)\leq 2T(|\mathcal{S}|,\kmax-1)+O((n+|\mathcal{S}|)\cdot |\mathcal{S}|^{1.5}+|\mathcal{S}|\cdot p_{\mathcal{S}}).
\]
By solving the above recursion, we have 
\[
T(|\mathcal{S}|,\kmax)=O((n+|\mathcal{S}|)\cdot |\mathcal{S}|^{1.5}\log_2 p_{\mathcal{S}}+|\mathcal{S}|\cdot p_{\mathcal{S}}\log_2 p_{\mathcal{S}}).\qedhere
\]
\end{proof}

\section{Evaluation}\label{sec:evaluation}
We present the evaluations of Algorithm \ref{alg:main} for 
time-sensitive networks to demonstrate optimality and high scalability. 

\vspace{1ex}
\noindent\textbf{Qualitative Evaluations.}  The qualitative evaluations are conducted on a real-life TSN system equipped on a metro train, where sensors across the train collect the status data of the train and then send data to controllers, which in turn send commands to actuators to keep the state of
the train close to the desired setpoint. The topology is described in Figure \ref{fig:train}. Here, the controller is located in Carriage 1, and most of the streams go from the other carriages to Carriage 1 or in the reverse direction. To verify the optimality of Algorithm \ref{alg:main}, we compared the schedules computed using Algorithm \ref{alg:main} to the ones generated by an SAT Modulo Theory (SMT) solver which exactly solves the SMT formulation of the corresponding problem \cite{survey}. Note that the SMT solver is poorly scalable and can only solve small instances with up to about 300 flows in 1 hour. 

We executed our evaluations in 10 practical scenarios on the topology in Figure \ref{fig:train} with respectively 3, 8, 40, 128, and 256 flows. We configured the SMT solver with a time limit of 1 hour for solving each scenario. Then we also run Algorithm \ref{alg:main} on the same instance. The evaluation shows that the SMT solver returns a no-wait scheduling if and only if Algorithm \ref{alg:main} returns one.
Besides, the average execution time of the SMT solver on these instances is about 3 minutes,  while Algorithm \ref{alg:main} solves these instances in 2 seconds on average.

\vspace{1ex}
\noindent\textbf{Scalability Evaluations.} 
To determine the scalability of Algorithm \ref{alg:main}, we measure the averaged execution times of Algorithm \ref{alg:main} with up to tens of thousands of flows on a personal computer (multi-processor machine (Intel(R) Core(TM) CPU i5-
12500H @ 2.50GHz) with 4P + 8E cores and 16 GB of memory). Precisely, we run Algorithm \ref{alg:main} on randomized instances with varying numbers of flows on the daisy-chain topology with varying lengths. The evaluation results show that Algorithm \ref{alg:main} can compute schedules for about 45,000 flows in only about 30 minutes. 

\section{Conclusion}\label{sec:conclusion}

In this paper, we propose an efficient algorithm to compute no-wait schedules on the daisy-chain topology, with a time complexity that scales polynomially in both the number of streams and the network size. The basic idea is to view the scheduling problem as a variant of graph coloring problem on the interval graphs. Our algorithm is proven to be optimal, and evaluations on real-life TSN systems justify the high scalability.  In our future work, we plan to investigate other commonly used topologies, such as the star topology, cycle topology, and tree topology.

\bibliography{ref}

\end{document}